\documentclass[letterpaper, 10 pt, conference]{ieeeconf}  

\usepackage[hidelinks]{hyperref}
\hypersetup{
    colorlinks=true,
    linkcolor=black,
    filecolor=magenta,      
    urlcolor=blue,
}
\usepackage{subfigure}
\usepackage[noadjust]{cite}
\usepackage{amsmath}
\usepackage{graphicx}
\usepackage{amssymb}
\usepackage{colortbl}
\usepackage{arydshln}
\usepackage{stfloats}
\usepackage{lipsum}
\usepackage[vlined,boxed,ruled,linesnumbered]{algorithm2e}
\SetKwProg{Fn}{Function}{}{}
\usepackage{algpseudocode}
\usepackage{multirow}
\usepackage{extarrows}
\usepackage{mathrsfs}
\usepackage{upgreek}

\newtheorem{remark}{Remark}

\newtheorem{problem}{Problem}

\newtheorem{assumption}{Assumption}
\newtheorem{proposition}{Proposition}
\newtheorem{definition}{Definition}

\definecolor{mygray}{gray}{0.8}
\usepackage[table]{xcolor}

\IEEEoverridecommandlockouts                              
\title{\LARGE \bf A Transition System Abstraction Framework for Neural Network Dynamical System Models}

\author{Yejiang Yang, Zihao Mo, Hoang-Dung Tran, and Weiming Xiang 
\thanks{This research was supported by the National Science Foundation, under NSF CAREER Award no. 2143351, NSF CNS Award no. 2223035, and NSF IIS Award nos. 2331937, 2331938.}
\thanks{Yejiang Yang, Zihao Mo, and Weiming Xiang are with the School of Computer and Cyber Sciences, Augusta University, Augusta GA 30912, USA. }
\thanks{Hoang-Dung Tran is with the School of Computing, University of Nebraska-Lincoln, Lincoln, NE 68588, USA.}
\thanks{Yejiang Yang is also with the School of Electrical Engineering, Southwest Jiaotong University, China.}
                }

\begin{document}

\maketitle
\thispagestyle{empty}
\pagestyle{empty}

\begin{abstract}
\boldmath
This paper proposes a transition system abstraction framework for neural network dynamical system models to enhance the model interpretability, with applications to complex dynamical systems such as human behavior learning and verification. To begin with, the localized working zone will be segmented into multiple localized partitions under the data-driven Maximum Entropy (ME) partitioning method. Then, the transition matrix will be obtained based on the set-valued reachability analysis of neural networks. Finally, applications to human handwriting dynamics learning and verification are given to validate our proposed abstraction framework, which demonstrates the advantages of enhancing the interpretability of the black-box model, i.e., our proposed framework is able to abstract a data-driven neural network model into a transition system, making the neural network model interpretable through verifying specifications described in Computational Tree Logic (CTL) languages. 
\end{abstract}

\section{Introduction}

In recent years, the popularity of machine learning models has surged dramatically, primarily owing to their remarkable capacity to offer valuable insights and accurate predictions, such as the remaining useful life prediction of lithium-ion batteries in \cite{WANG2023108920}, controlling physical human-robot interaction in \cite{liu2023fuzzy}, modeling the thermal conductivity of water-based nanofluids containing magnetic copper nanoparticles in \cite{ghazvini2020experimental}. Nonetheless, conventional machine learning models are often regarded as ``black boxes", making the control and verification of neural network models relying on real-time trajectory computation with tools such as those proposed in \cite{zhang2018efficient}, \cite{wang2021beta,tran2020nnv,fazlyab2020safety} and \cite{yang2022guaranteed}, etc., which are often considered computationally expensive and challenging for real-time computation. Additionally, the dilemma of black-box verification limits the applications of neural networks in dynamical systems modeling, thus, the need for an interpretable abstraction for neural network dynamical system that unveils the black box is critical.

Thanks to recent works on hybrid system learning in \cite{poli2021neural}, \cite{sprague2022efficient}, \cite{beg2017model} and research on learning frameworks via distributed learning structures in \cite{chowdhery2022palm}, \cite{jacobs1991adaptive} and \cite{horvath2023stochastic}, etc., it is possible to break down one general system into multiple sections and the relationships between them. However, most of the existing works tend to concentrate on modeling hybrid dynamics using hybrid models, inadvertently overlooking the essential benefits of employing distributed learning models, which can significantly enhance the interpretability \cite{li2022interpretable} of the learning model while also broadening the scope of model verification, e.g., verifying the learning model with user-specified properties via Linear Temporal Logic \cite{liu2013synthesis}. 

Diverging from the approach of employing multiple neural networks to learn hybrid systems, which primarily emphasizes subsystem learning, the challenge of investigating inter-subsystem relationships is rooted in the abstraction of the model. To begin with, since prior knowledge of a system is often limited, especially in the context of data-driven modeling, we introduce a transition system abstraction approach. This method involves partitioning the model's working zone using techniques founded on Maximum Entropy (ME) partitioning, as discussed in \cite{yang2023data}. Following this, the transitions between distinct partitions will be approximated via reachability analysis, which enables us to abstract the neural network dynamical system and the partition relationship into a transition system. Consequently, this method provides an intuitive means of discovering the relationship between subsections of a neural network and contributes to its improved interpretability. 

The main contributions of this paper can be summarized as follows.
\begin{itemize}
    \item The proposed framework involves the introduction of a novel transition system abstraction for neural network dynamical systems. This model serves the dual purpose of verifying the behaviors of localized subsystems and shedding light on their interrelationships by estimating transitions within the neural network dynamical system's working zone through the application of reachability analysis, a technique uniquely suited to analyzing localized subsystem interactions. 
    \item The proposed neural transition system can be established on an existing neural network dynamical system to make the black-box model interpretable. The specifications can be verified through Computational Tree Logic (CTL), thereby broadening the possibilities for neural-network-based model verification. The approach is validated by human handwriting dynamics learning and verification.
\end{itemize}

The paper is structured as follows: Preliminaries and problem formulations are given in Section II. The main result, a data-driven neural transition system abstraction-based framework, is given in Section III. In Section IV, the applications for abstracting the learning model of the LASA handwriting data set with our proposed framework are presented. Conclusions are given in Section V.  

\emph{Notations}: In the rest of the paper, $\mathbb{N}$ denotes the natural number sets, $\mathbb{R}$ is the field of real numbers; $\mathbb{R}^n$ stands for the vector space of $n$-tuples of real number; $\underline{X}$ and $\overline{X}$ are the lower bound and upper bound of an interval $X$, respectively.

\section{Preliminaries and Problem Formulation}
A neural network dynamical system model describes the dynamical system behaviors through an input-output point of view in the form of
\begin{align}\label{eq: neural network system}
    x(k+1)=\Phi(x(k),u(k))
\end{align}
in which $\Phi:\mathbb{R}^{n_x+n_u}\to\mathbb{R}^{n_x}$ is the trained neural network. A neural network dynamical system model is trained based on samples, which can be well-trained under an adequate sample set. Therefore, we assume that the $\Phi$ applies only to a localized working zone $\mathcal{X}$ which can be obtained based on samples.  

Neural network $\Phi$ can approximate dynamical systems with high accuracy, but the lack of interpretability burdens the processes of verification and control synthesis. In this paper, we aim to develop a transition system abstraction for a neural network dynamical system model defined as 

\begin{definition}\label{def_transition system}
A transition system abstraction is a tuple $\mathcal{T} \triangleq \left \langle \mathcal{X},P,\mathcal{E}\  \right \rangle$ consisting of the following components:
\begin{itemize}

\item \emph{Working zone} $\mathcal{X}\subset \mathbb{R}^{n_x} $ is the working zone that users are paying attention to. 

\item \emph{Partitions}: $P=\{P_1,\cdots, P_n\}$ is a finite set of partitions, which are subspaces of $\mathcal{X}$ where: 1) $P_i\subset \mathcal{X}$; 2) $\mathcal{X}=\bigcup{P_i} $; 3) $P_i\bigcap P_j=\emptyset$. 

\item \emph{Transitions}: $\mathcal{E}\subset \mathbb{B}^{n\times n}$ is the transition matrix where $e_{ij}$ is the $i$th line, $j$th row boolean variable of transition matrix that describes whether there exists a transition from $i$th partition to $j$th partition. For example, a transition from $i$th partition to $j$th partition is taking place when $x(k)\in P_i$ and $x(k+1) \in P_j$, $e_{ij}=1\in \mathcal{E}$, if there is no transition, $e_{ij}=0$.
\end{itemize}
\end{definition}


In this section, we will be presenting the preliminaries and problem formulation in this paper, which includes the variation of the Shannon Entropy, feed-forward neural network, and its reachability analysis.

\subsection{Variation of the Shannon Entropy}

Learning models are usually established under limited prior knowledge, which means the system information is provided mostly by samples. The Shannon Entropy measures the uncertainty or disorder in a given data distribution while only relying on samples. Specifically, according to \cite{rajagopalan2006symbolic} the Shannon Entropy of the dynamics with $l\in\mathbb{R}$ partitions (denoted as $H(l)$) is 
\begin{align}\label{ali: Shannon Entropy}
     H(l)=-\sum_{i=1}^l p_i \log_2{p_i},
\end{align}
where $p_i$ denotes the probability of the $i$th partition of the dynamics. If we partition the dynamics into $k\ge l$ local partitions, we can compute the variation of the Shannon Entropy based on (\ref{ali: Shannon Entropy}) in
\begin{align}
    \Delta H= H(k)-H(l)
\end{align}
where $\Delta H$ reviles the variation of the Shannon Entropy, which helps in identifying the difference between different partitioning sets. We denote the probability of one partition in the form of
\begin{align}
    p_i=\frac{N_i}{N},
\end{align}
in which $p_i$ is the occurrence probability of partition $\mathcal{P}_i$ in (\ref{ali: Shannon Entropy}), $N_i$ is the number of samples that are within $\mathcal{P}_i$, and $N$ is the number of all samples. 

\subsection{Neural Network Dynamical System Model and Reachability Analysis}

Our proposed framework focuses on developing an interpretable neural transition model on an existing trained neural network dynamical system model in the form of (\ref{eq: neural network system}). Among all the learning structures of neural networks, feed-forward neural networks have received particular attention for their simple structure and varied training approaches. In this section, We will take a feed-forward neural network as an example to investigate the internal propagation mechanism of a neural network dynamical system model.

For a feed-forward neural network $\Phi:\mathbb{R}^{n_0}\to\mathbb{R}^{n_L}$ with $L$ layers, the process of inter-layer propagation is in the form of
\begin{align}\label{eq:feed-forward neural network}
    x_{k+1}=h_k(W_{k}x_k+b_k)
\end{align}
in which $x_k\in\mathbb{R}^{n_k}$ is the output of $k$th layer, $W_{k}\in\mathbb{R}^{n_{k+1}}\times \mathbb{R}^{n_{k}}$ is the weight matrix that connects $k$th layer to $k+1$th layer, $b_k\in\mathbb{R}^{n_k}$, $h_k:\mathbb{R}^{n_k} \to \mathbb{R}^{n_k}$ is the activation function of $k$th layer.

Based on (\ref{eq:feed-forward neural network}), namely the weight and bias are given, we will be able to compute the reachable set output defined as follows.
\begin{definition}\label{def: reachable set analysis}
    For a neural network dynamical system model $\Phi:\mathbb{R}^{m}\to\mathbb{R}^n$, when given a set input $\mathcal{X}\subset\mathbb{R}^m$, we define 
    \begin{align}
        \mathcal{Y}=\{y\in\mathbb{R}^n\mid y = \Phi(x),\forall x\in \mathcal{X}\},
    \end{align}
as the reachable set output of the neural network $\Phi$ when given input set $\mathcal{X}$.
\end{definition}   

The reachable set defined in Definition \ref{def: reachable set analysis} facilitates verification through the input-output set, which presents a unique angle to tackle the challenge of black-box verification.

\begin{definition} \label{def_reachable set estimation}
Given initial state set $\mathcal{X}_{(0)}$ and input set $\mathcal{U}$, the reachable set at $k$th time step $\mathcal{X}_{(k)}$ of a neural network dynamical system model $\Phi$ is defined as 
\begin{align}\label{equ: reachable set analysis}
   \mathcal{X}_{(k)} =\left\{x(k;x_0,u(\cdot))\in\mathbb{R}^{n} \mid x_0 \in \mathcal{X}_{(0)},~ u(k) \in \mathcal{U}\right\},
\end{align}
and the union of $\mathcal{X}_{(k)}$ over $[0,K]$ defined by
\begin{equation}\label{eq: entropy variation}
\mathcal{R}_{(K)}=\bigcup_{k=0}^{K}\nolimits\mathcal{X}_{(k)}, 
\end{equation}
is the reachable set over time interval $[0,K]$. 
\end{definition}

\begin{remark}
Instead of relying on the input-output point of view, reachability analysis makes the set-valued computation possible so that we can handle the perturbations and inaccuracies in safety-critical scenarios. However, the reachability analysis heavily relies on iterations of the reachable set computation which can be very computationally expensive and time-consuming. 
\end{remark}

\subsection{Problem Formulation}
Neural networks are often considered black-box models, primarily due to their limited interpretability. Consequently, this presents challenges for real-time verification, as available methods are constrained and computational costs are high. 
\begin{problem}\label{problem_1}
Based on a neural network dynamical system $\Phi$, how can we abstract one general $\Phi$ into a transition system $\mathcal{T}$ that makes it possible to avoid computational-expensive real-time verification while enhancing the interpretability of a neural network dynamical system? 
\end{problem}

In the rest of the paper, we will focus on solving Problem \ref{problem_1} in detail.

\section{Transition System Abstraction via ME Partitioning and Reachability Analysis}


We aim to solve Problem \ref{problem_1} with our proposed transition system abstraction while computing its transitions with set-valued reachability analysis to generate the abstracted transition graph to reveal its inner transitions.

To begin with, the neural network dynamical system will be abstracted within the working zone. The assumption of the working zone is given as follows.

\begin{assumption}
The working zone of neural network dynamical system $\Phi$ is within the localized state space $x\in\mathcal{X}$, given the external input bound where $u\in[\underline{u},\overline{u}]$. 
\end{assumption}

\begin{remark}
A neural network dynamical system is trained based on samples, which can be well-trained under an adequate sample set. Therefore, we assume that the learning model applies only to a localized working zone based on samples. This localized approach allows for better abstraction of the dynamics within that specific region. 
\end{remark}




\subsection{ME Partitioning}
Based on the working zone of the neural dynamical system, we are able to generate numerous traces that are included in the set of samples when given random initial starting points with (\ref{eq: neural network system}).

\begin{definition}\label{def_sampled data}
Set of samples $\mathcal{W}=\{w_1,w_2,\cdots,w_L\}$ of neural network dynamical system (\ref{eq: neural network system}) is a collection of sampled $L$ traces obtained by measurement, where for each trace $w_i$, $i = 1,\ldots,L$, is a finite sequence of time steps and data $(k_{0,i},d_{0,i}),(k_{1,i},d_{1,i}),\cdots,(k_{M_i,i},d_{M_i,i})$ in which
\begin{itemize}
    \item $k_{\ell,i}\in(0,\infty)$ and $k_{\ell+1,i}=k_{\ell,i}+1$, $\forall \ell = 0,1,\ldots,M_i,\ \forall i = 1,2,\ldots,L$.
    \item $d_{\ell,i}=[x^\top_{i}(k_{\ell,i}),~u^\top_{i}(k_{\ell,i})]^\top \in \mathbb{R}^{n_x+n_u}$,
    $\forall \ell = 0,1,\ldots,M_i$, $\forall i = 1,2,\ldots,L$, where $x_{i}(k_{\ell,i}),u_{i}(k_{\ell,i})$ denote the state and input of the system at $\ell$th step for $i$th trace, respectively.
\end{itemize}
\end{definition}

In this section, we will be presenting the ME partitioning technique that divides one general working zone $\mathcal{X}$ into a set of partitions $\{P_1,\ldots, P_n\}$ based on the set of samples, while computing the transitions between partitions with reachable set computation.

First, ME partitioning will partition the working zone into multiple partitions based on the Shannon Entropy of the abstraction. Based on (\ref{eq: entropy variation}), it takes the samples within one partition for Shannon Entropy computation. We define the segment data set, which is subsequently defined under this process, as follows.

\begin{definition}\label{def_segment set}
Given a set of samples $\mathcal{W}$ of neural network dynamical system $\Phi$ and a set of partitions $P=\{P_1,\ldots,P_N\}$, the segment data set is 
    \begin{align}
        X=\{X_1,X_{2},\ldots,X_{N}\}
    \end{align}
in which
\begin{align}\label{equ:segment sample}
     X_q \triangleq \{x_{ i}(k_{ \ell, i}) \mid x_{ i}(k_{ \ell, i})\in P_q\}
\end{align}
in which $\forall i=1,2,\ldots,L$, $\forall \ell=0,1,\ldots,M_i$, $x_i(k_{\ell,i})$ is any sampled state $\forall q= 1,\ldots,N$ in traces $w_i$. Therefore, $X_q$ contains all the sampled state traces evolving within partition $P_q$. 
\end{definition}

\begin{remark}
In this paper, the working zone $\mathcal{X}$ is defined by hyper-rectangles which are given as follows: For any bounded state set $\mathcal{X} \subseteq \mathbb{R}^{n_x}$, we define $\mathcal{X} \subseteq \bar{\mathcal{X}}$, where $\bar{\mathcal{X}} = \{x \in \mathbb{R}^{n_x} \mid \underline{x} \le x \le \bar{x}\}$, in which $\underline{x}$ and $\bar{x}$ are defined as the lower and upper bounds of  $x$ in $\mathcal{X}$ as  $\underline{x}=[\mathrm{inf}_{x\in \mathcal{X}}(x_1),\ldots,\mathrm{inf}_{x\in \mathcal{X}}(x_{n_x})]^{\top}$ and $\bar{x}=[\mathrm{sup}_{x\in \mathcal{X}}(x_1),\ldots,\mathrm{sup}_{x\in \mathcal{X}}(x_{n_x})]^{\top}$, respectively. We denote the distance of $\mathcal{X}$ at $k$th dimension as $d(\mathcal{X}_k)\triangleq (\sup_{x\in\mathcal{X}}(x_k)-\mathrm{inf}_{x\in\mathcal{X}}(x_k))$.
\end{remark}

The goal of the Maximum Entropy Partitioning Algorithm is to determine the optimal set of partitions. Specifically, we bisect one partition into two once at a time, e.g., $P_i\to \{P_{i,1},P_{i,2}\}$ will lead to the variation of the Shannon Entropy
\begin{align}\label{ali_shannon bisecting}
     \Delta H_i=\frac{N_{i,1}\log_2{\frac{N_{i,1}+N_{i,2}}{N_{i,1}}}+N_{i,2}\log_2{\frac{N_{i,1}+N_{i,2}}{N_{i,2}}}}{N_i}
\end{align}
where $N_{i,1}=\left|X_{i,1}\right|$ is the sample number of $P_{i,1}$ under (\ref{equ:segment sample}), etc.

Based on (\ref{ali_shannon bisecting}), the Shannon entropy will always increase due to $\Delta H_i>0$ in (\ref{ali_shannon bisecting}) always holds. To provide a proper partition set, we will set a threshold $entropy>0$ for $\Delta H_i$ as the stop condition, namely the bisection process will stop if $\Delta H< entropy$. The detailed ME partitioning is given in pseudo-code in Algorithm \ref{alg1} in which we set a minimum threshold $L_{min}$ for the length of the partitions.

\begin{algorithm}[t!] 
\SetAlgoLined
\SetKwInOut{Input}{Input}
\SetKwInOut{Output}{Output}
\Input{Set of samples $\mathcal{W}$; Working Zone $\mathcal{X}$; Threshold for Variation of The Shannon Entropy $entropy$; Threshold for minimum length of the partitions $L_{min}$.}
\Output{Set of Partitions $\{P_1,P_2,\ldots,P_M\}$.}
$P_{save} \gets \emptyset$;~ $\mathcal{W}_{save} \gets \emptyset$;\\
$P_1 \gets \mathcal{X} $;~$r \gets \dim{\mathcal{X}}$; $Length \gets \max(d(\mathcal{X}))$;\\
Obtain $X_1$ from $\mathcal{W}$;~$q \gets |X_1|$; \\

\While{$ Length \ge L_{min}$}{
    $[i,Length,k] \gets$  $\max(d(P))$\\
    Obtain $P_{temp1}$ and $P_{temp2}$ by bisecting $P_i$ on dimension $k$\\
    Obtain $X_{temp1}$ and $X_{temp2}$ based on $P_{temp1}$ and $P_{temp2}$ \\
    \eIf{the variation of the Shannon Entropy $\Delta H_i\ge entropy$}{
   $P_i \gets \{P_{temp1},P_{temp2}\}$\\
 $X_{i}\gets\{X_{temp1},X_{temp2}\}$}
    {
    Add $P_{i}$ to $P_{save}$ and delete $P_i$\\
    Add $\mathcal{W}_{i}$ to $\mathcal{W}_{save}$ and delete $\mathcal{W}_i$
    }
}
$P\gets\{P,P_{save}\}$\\
\Return{$\{P_1,P_2,\ldots,P_M\}$.}
\caption{Pseudo Code for ME Partitioning}\label{alg1} 
\end{algorithm}

\subsection{Transition Computation via Reachability Analysis}
After ME partitioning, we will be able to compute the transition relationship between partitions through the reachability analysis. According to (\ref{equ: reachable set analysis}), we apply the reachable set computation technique in the form of
 \begin{align}\label{eq:reachable set estimation}
     \mathcal{Y}=[\Phi](\mathcal{X})
 \end{align}
in which $[\Phi]$ is a reachable set computation tool such as the Neural Network Verification (NNV) toolbox in \cite{tran2020nnv}, the set-valued reachability analysis in \cite{xiang2018output}, Beta-crown in \cite{wang2021beta}, etc. 

When given the input of one partition $P_i$, we will be able to detect whether the reachable set output of the neural dynamical system intersects with any partitions (included itself), namely, given input set $\mathcal{U}$ if $[\Phi](P_i,\mathcal{U})\cap P_j \neq \emptyset$, then $e_{ij}=1$, based on which the transitions can be obtained. The process of transition computation is given in Algorithm \ref{alg2}.

\begin{algorithm}[t!] 
\SetAlgoLined
\SetKwInOut{Input}{Input}
\SetKwInOut{Output}{Output}
\Input{Neural Network Dynamical System $\Phi(\cdot)$; Set of Partitions $P$.}
\Output{Set of Transitions $\mathcal{E}$.}
Get the number of partitions $n$ from $P$\\
$i\gets1$; $j\gets1$;\\
\While{$ i \le n$}{
   Computing $[\Phi](P_i)$ using (\ref{eq:reachable set estimation})\\
    \While{$j\le N$}{\eIf{$[\Phi](P_i)\cap P_j\neq\emptyset$}{$e_{ij}=1$}
    {$e_{ij}=0$}{$j\gets j+1$}
    }
    $i\gets i+1 $
    }
\Return{$\mathcal{E}=\begin{bmatrix}
    e_{11}  &\cdots &e_{1n}\\
    \vdots  &\ddots &\vdots\\
    e_{n1}  &\cdots &e_{nn}
\end{bmatrix}$.}
\caption{Pseudo Code for Transition Computation}\label{alg2} 
\end{algorithm}

According to the transition system in Definition \ref{def_transition system}, $e_{ii}=1$,  if $[\Phi](P_i,\mathcal{U})\cap P_i\neq \emptyset$, this will result in the self-loop of transition system abstraction. From the verification point of view, $e_{ii}=1$ indicates two scenarios:
\begin{enumerate}
    \item $P_i$ contains invariant set for neural network dynamical system $\Phi$ given the external set $\mathcal{U}$. 
    \item $P_i$ does not contain an invariant set, $e_{ii}=1$ only because sampling time being too short, i.e., $P_i$ does not contain an invariant set when $\exists x(k_{l,i})\in P_i, \ x(k_{l+1,i})\in P_i$.
\end{enumerate}

Scenario 2 may lead to potential misjudgments when we apply the verification technique such as CTL in verification \cite{lahijanian2015formal}. Therefore, we should take $e_{ii}=0$ under certain conditions. In practice, we reduce the self-loop under Scenario 2 with the help of the traces that are segmented by partitions.
\begin{definition}\label{def_segment trace}
    Given trace $w_i$, partition $P_q$, the segment of $w_i$ is 
    \begin{align}
        X_{i,q}=\{x_i(k_{ \ell,i})\mid x_i(k_{ \ell,i})\in P_q\},\ \forall \ell=0,1, \ldots,M_i
    \end{align}
    in which $P_q$ and $w_i$ are the given partition and trace.
\end{definition}

Under Definition \ref{def_segment trace}, we can reduce the self-loop on Scenario 2 with the following assumption.

\begin{assumption}\label{ass_invariant set}
Given threshold $n\in\mathbb{N}$, if there is no invariant set within $P_q$ when given external input set $\mathcal{U}$, then $ |X_{ i,q}|\le n, \forall i=1,\ldots,L$.    
\end{assumption}

Assumption \ref{ass_invariant set} can be very helpful in simplifying the self-loop under Scenario 2 when given adequate neural network dynamical system information.

Neural network dynamical systems $\Phi$ are often trained as the approximation of the system's ideal dynamical system description $f$ in the form of
\begin{align}
x(k+1)=f(x(k),u(k)).    
\end{align}

To make sure we have an accurate abstraction for the dynamical system, we can provide guaranteed transition computation, which means our transition abstraction can still be applied for $f$ with the following proposition.
\begin{proposition}\label{proposition_1}
Given an output reachable set estimation in $\mathcal{X}'=[\Phi](P_i)$ and $\mathcal{X}'\subset P_j$, $\dim (\mathcal{X}')=m$, the estimation for transitions of the neural transition system will be the same as the dynamics $f$, there is a transition in neural transition system as is in the ideal dynamical system description $e_{ij=1}$, if:
(1) $\underline{x}'_{k}-\epsilon\ge\underline{p}_{j,k}$, (2) $\overline{x}'_{k}+\epsilon\le\overline{p}_{j,k}$, $\forall k\le m$, where $\epsilon\ge0$ is the guaranteed distance between dynamics $f$ and its neural network approximation $\Phi$ in \cite{yang2022guaranteed}, $\underline{x}'_k$ is the lower bound of $\mathcal{X}'$ at $k$th dimension, etc.
\end{proposition}
\begin{proof}
    To begin with, the guaranteed distance (Definition in \cite{yang2022guaranteed}) is
    \begin{align*}
        \left\|\Phi(x)-f(x)\right\|\le \epsilon,
    \end{align*}
    as for the neural transition systems' successive $x'\in P_j$, it yields $\left\|x'-f(x)\right\|\le \epsilon$. It leads to
    \begin{align}\label{ali_proof1}
        \begin{cases}
            \overline{x}'_k\ge \overline{f}_k(x)-\epsilon\\
            \underline{x}'_k \le \underline{f}_k(x)+\epsilon
        \end{cases}
        \forall k\le m,
    \end{align}
    in which $\overline{f}_k(x)$ is the $k$th dimension upper bound of the output set for $f$, etc. If the successive set $\mathcal{X}'\subset P_j$, then $\Phi(x)\in P_j,\  \forall x\in P_i$,
    which leads to
    \begin{align}\label{ali_proof2}
        \begin{cases}
            \underline{x}'_k-\epsilon\ge\underline{p}_{j,k}\\
            \overline{x}'_k+\epsilon\le\overline{p}_{j,k}
        \end{cases}
        \forall k\le m.
    \end{align}
    The proof is complete by combing (\ref{ali_proof1}) with (\ref{ali_proof2}), which shows $f(x)\in P_j,\ \forall x\in P_i$, indicating there is a transition in $\mathcal{T}$ and the ideal dynamical system description $f$. 
\end{proof}

Transition system abstraction for a neural network dynamical system can be summarized as follows.
\begin{itemize}
    \item The localized working zone, i.e., partitions can be obtained based on an ME partitioning process, which is completely data-driven and can be easily tuned by adjusting the threshold.  
    \item The transitions can be off-line computed by reachability analysis. To provide accurate transition estimations between the $\Phi$ and $f$, the guaranteed distance is introduced in the proposed Proposition \ref{proposition_1} for our abstraction system.
\end{itemize}

After ME partitioning and transition computation, we will be able to abstract the neural network dynamical system $\Phi$ into a transition system abstraction $\mathcal{T}=\left \langle\mathcal{X},P,\mathcal{E}\right\rangle$.

\section{Application to Learning and Verification on Human Handwriting Dynamics}
Learning-based methods have been promoted as an effective way to model motions \cite{reinhart2017hybrid,reinhart2011neural}. Consider the complex dynamical system such as human behaviors borrowed from the LASA data set \cite{LASA} which contains handwriting motion demonstrations of human users of $30$ shapes. The challenges of modeling this kind of complex dynamical system can be summarized as follows.
\begin{itemize}
    \item The lack of understanding of human complex behavior. Limited samples may lead to a deep neural network dynamical system model not getting well-trained.
    \item Human demonstrations contain abrupt changes, which means the behaviors of the trained neural network dynamical system model will have differences in local areas of the working zone.
    \item We have limited methods to verify the neural network dynamical system model for its lack of interpretability. 
\end{itemize}


Due to the specifics above, we will demonstrate our proposed abstraction-based framework based on the neural network dynamical system approximation of the Human demonstrations. This process can be summarized as follows\footnote{The code for the human handwriting application and transition system abstraction tool is available at: \url{https://github.com/aicpslab/Abstraction-Toolbox-for-Neural-Transition-Systems}}.
\begin{itemize}
    \item We train $20$-ReLU activated neurons ELMs as the neural dynamical system model for learning Angle and P shape handwriting demonstrations. By setting the maximum and minimum values of the sample data, we obtain the working zone of neural network dynamical systems.
    \item A threshold $entropy=4\times e^{-2}$ is set as the variation of maximum entropy in Algorithm \ref{alg1}, which results in $28$ partitions in transition system abstraction for Angle shape Fig. \ref{Fig_sample&partitions} (a) and $27$ partitions in learning model for P shape Fig. \ref{Fig_sample&partitions} (b). We partition the working zone based on the randomly generated $20$ trajectories.
    \item Based on the partitions of the neural transition systems, we employ the neural verification toolbox for transition computations in Algorithm \ref{alg2}, and reduce the self-loop with $n=100$ in Assumption \ref{ass_invariant set}. Then we will be able to provide the transition graphs in \ref{Fig_partitions} (a) and (b). From Fig. \ref{Fig_sample&partitions} to Fig. \ref{Fig_partitions}, the transition relationships between local working zones are able to be interpreted through our transition system abstraction.
    \item We verify our transition system abstraction with the Computation Tree Logic (CTL) formulae \cite{pan2016model}, in which $\diamond$ or $\Box$ denote the for $some$ or $all$ traces, $\bigcirc$ denotes for next step. The results are given in Tab \ref{tab5}, in which $\phi_1$ is true means there exists a trace from $P_{28}$ to $P_{12}$, $\phi_2$ is true for $\mathcal{T}_{Angle}$ means for all traces that start from $P_{28}$, it will go to $P_{25}$ as the next partitions, $\phi_3$ means for $\mathcal{T}_{Pshape}$ there exists a trace in $P_4$ then will be in $P_{16}$ next.
\end{itemize}

\begin{figure}[htbp]
\centering
	\subfigure[$|P_{Angle}|=28$.]{\includegraphics[width=4.2cm]{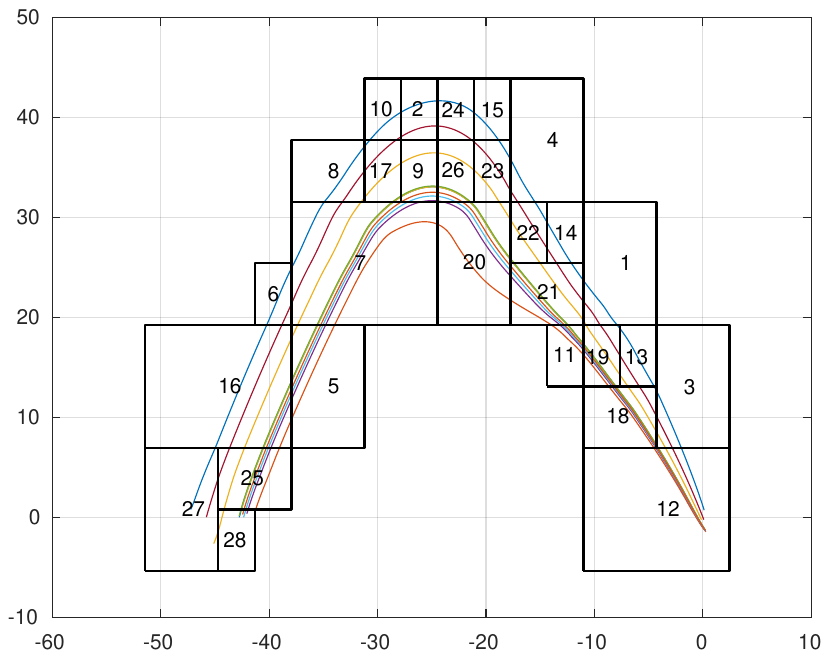}}
	\subfigure[$|P_{P shape}|=27$.]{\includegraphics[width=4.2cm]{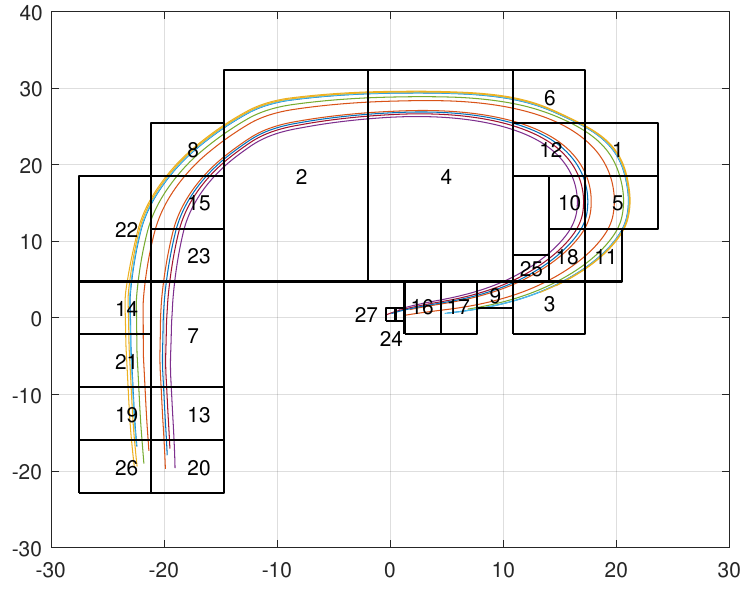}}
	\caption{There are $28$ partitions from (a) while $27$ partitions from (b), information-rich areas are allocated more partitions.}
	\label{Fig_sample&partitions} 
\end{figure}

 \begin{figure}[htbp]
\centering
	\subfigure[Transition system abstraction $\mathcal{T}_{Angle}$.]{\includegraphics[width=6cm]{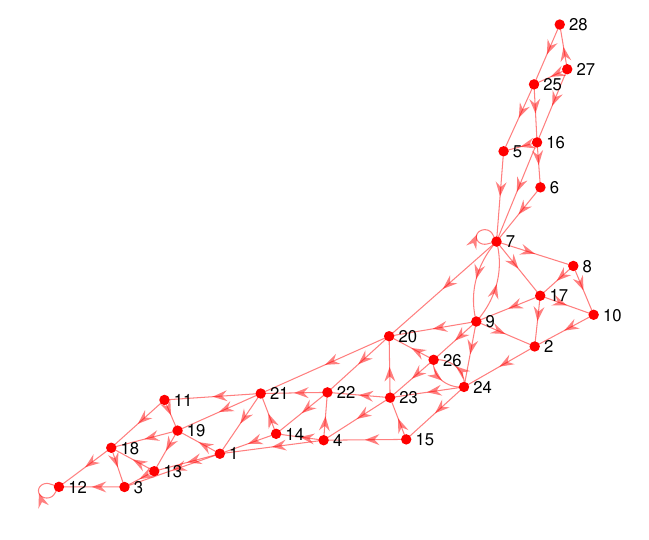}}
	\subfigure[Transition system abstraction $\mathcal{T}_{Pshape}$.]{\includegraphics[width=6cm]{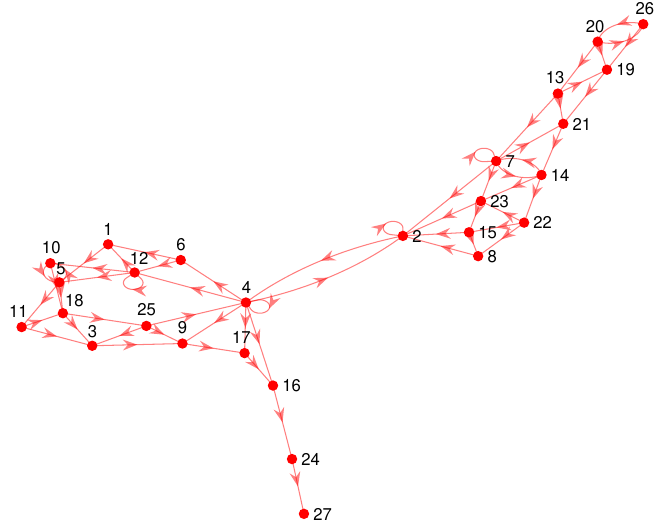}}
	\caption{The transition relationships are denoted by red arrows. Based on (a) and (b) we can clearly see the model's working flow, which means the learning model's interpretability is enhanced, such as results in Table \ref{tab5} where complex CTL formulae can be verified.}
	\label{Fig_partitions} 
\end{figure}

\begin{table}[t!]
	\centering
	\caption{Verification results of $CTL$ formula:  $\mathcal{T}_{Angle}$ with $P_{27}$ as the initial partition, and $\mathcal{T}_{P shape}$ with $P_{26}$ as the initial partition. }\label{pro_1_tab}
	\label{tab5}
	\begin{tabular}{|ccc|}
		\hline 
        \textbf{$CTL$ formula}  &\textbf{$\mathcal{T}_{Angle}$} &\textbf{$\mathcal{T}_{Pshape}$} \\
		\hline\hline
$\phi_1=\diamond\exists P_{12}$   &$true$ &$-$ \\
$\phi_2=\Box \bigcirc P_{25}$   &$true$  &$-$\\
$\phi_3=\diamond\exists ( P_4 U P_{16})$   &$-$ & $true$ \\
\hline
	\end{tabular}
\end{table} 
 
\section{Conclusions}
In this paper, a transition system abstraction-based framework for the neural network dynamical system model is given to promote interpretability. First, the working zone of the transition system will be obtained based on the samples. Then, we employ the ME partitioning to bisect the working zone of the neural network dynamical system model into multiple partitions, while the transitions of the learning model will be computed under reachable set computations. Transition system-based abstractions and verification for neural dynamical systems of Angle shape and P-shape are presented to illustrate the effectiveness, which indicates the interpretability is enhanced while promoting a new way of neural network dynamical system model verification.

\bibliographystyle{ieeetr}
\bibliography{ref}

\begin{thebibliography}{10}

\bibitem{WANG2023108920}
S.~Wang, Y.~Fan, S.~Jin, P.~Takyi-Aninakwa, and C.~Fernandez, ``Improved anti-noise adaptive long short-term memory neural network modeling for the robust remaining useful life prediction of lithium-ion batteries,'' {\em Reliability Engineering \& System Safety}, vol.~230, p.~108920, 2023.

\bibitem{liu2023fuzzy}
A.~Liu, T.~Chen, H.~Zhu, M.~Fu, and J.~Xu, ``Fuzzy variable impedance-based adaptive neural network control in physical human--robot interaction,'' {\em Proceedings of the Institution of Mechanical Engineers, Part I: Journal of Systems and Control Engineering}, vol.~237, no.~2, pp.~220--230, 2023.

\bibitem{ghazvini2020experimental}
M.~Ghazvini, H.~Maddah, R.~Peymanfar, M.~H. Ahmadi, and R.~Kumar, ``Experimental evaluation and artificial neural network modeling of thermal conductivity of water based nanofluid containing magnetic copper nanoparticles,'' {\em Physica A: Statistical Mechanics and its Applications}, vol.~551, p.~124127, 2020.

\bibitem{zhang2018efficient}
H.~Zhang, T.-W. Weng, P.-Y. Chen, C.-J. Hsieh, and L.~Daniel, ``Efficient neural network robustness certification with general activation functions,'' {\em Advances in Neural Information Processing Systems}, vol.~31, 2018.

\bibitem{wang2021beta}
S.~Wang, H.~Zhang, K.~Xu, X.~Lin, S.~Jana, C.-J. Hsieh, and J.~Z. Kolter, ``Beta-crown: Efficient bound propagation with per-neuron split constraints for neural network robustness verification,'' {\em Advances in Neural Information Processing Systems}, vol.~34, 2021.

\bibitem{tran2020nnv}
H.-D. Tran, X.~Yang, D.~Manzanas~Lopez, P.~Musau, L.~V. Nguyen, W.~Xiang, S.~Bak, and T.~T. Johnson, ``\textsc{NNV}: the neural network verification tool for deep neural networks and learning-enabled cyber-physical systems,'' in {\em International Conference on Computer Aided Verification}, pp.~3--17, Springer, 2020.

\bibitem{fazlyab2020safety}
M.~Fazlyab, M.~Morari, and G.~J. Pappas, ``Safety verification and robustness analysis of neural networks via quadratic constraints and semidefinite programming,'' {\em IEEE Transactions on Automatic Control}, vol.~67, no.~1, pp.~1--15, 2020.

\bibitem{yang2022guaranteed}
Y.~Yang, T.~Wang, J.~P. Woolard, and W.~Xiang, ``Guaranteed approximation error estimation of neural networks and model modification,'' {\em Neural Networks}, vol.~151, pp.~61--69, 2022.

\bibitem{poli2021neural}
M.~Poli, S.~Massaroli, L.~Scimeca, S.~Chun, S.~J. Oh, A.~Yamashita, H.~Asama, J.~Park, and A.~Garg, ``Neural hybrid automata: Learning dynamics with multiple modes and stochastic transitions,'' {\em Advances in Neural Information Processing Systems}, vol.~34, pp.~9977--9989, 2021.

\bibitem{sprague2022efficient}
C.~Sprague, {\em Efficient and Trustworthy Artificial Intelligence for Critical Robotic Systems}.
\newblock PhD thesis, Kungliga Tekniska h{\"o}gskolan, 2022.

\bibitem{beg2017model}
O.~A. Beg, H.~Abbas, T.~T. Johnson, and A.~Davoudi, ``Model validation of \textsc{PWM DC-DC} converters,'' {\em IEEE Transactions on Industrial Electronics}, vol.~64, no.~9, pp.~7049--7059, 2017.

\bibitem{chowdhery2022palm}
A.~Chowdhery, S.~Narang, J.~Devlin, M.~Bosma, G.~Mishra, A.~Roberts, P.~Barham, H.~W. Chung, C.~Sutton, S.~Gehrmann, {\em et~al.}, ``Palm: Scaling language modeling with pathways,'' {\em arXiv preprint arXiv:2204.02311}, 2022.

\bibitem{jacobs1991adaptive}
R.~A. Jacobs, M.~I. Jordan, S.~J. Nowlan, and G.~E. Hinton, ``Adaptive mixtures of local experts,'' {\em Neural Computation}, vol.~3, no.~1, pp.~79--87, 1991.

\bibitem{horvath2023stochastic}
S.~Horv{\'a}th, D.~Kovalev, K.~Mishchenko, P.~Richt{\'a}rik, and S.~Stich, ``Stochastic distributed learning with gradient quantization and double-variance reduction,'' {\em Optimization Methods and Software}, vol.~38, no.~1, pp.~91--106, 2023.

\bibitem{li2022interpretable}
X.~Li, H.~Xiong, X.~Li, X.~Wu, X.~Zhang, J.~Liu, J.~Bian, and D.~Dou, ``Interpretable deep learning: Interpretation, interpretability, trustworthiness, and beyond,'' {\em Knowledge and Information Systems}, vol.~64, no.~12, pp.~3197--3234, 2022.

\bibitem{liu2013synthesis}
J.~Liu, N.~Ozay, U.~Topcu, and R.~M. Murray, ``Synthesis of reactive switching protocols from temporal logic specifications,'' {\em IEEE Transactions on Automatic Control}, vol.~58, no.~7, pp.~1771--1785, 2013.

\bibitem{yang2023data}
Y.~Yang, Z.~Mo, and W.~Xiang, ``A data-driven hybrid automaton framework to modeling complex dynamical systems,'' in {\em 2023 IEEE International Conference on Industrial Technology (ICIT)}, pp.~1--6, IEEE, 2023.

\bibitem{rajagopalan2006symbolic}
V.~Rajagopalan and A.~Ray, ``Symbolic time series analysis via wavelet-based partitioning,'' {\em Signal processing}, vol.~86, no.~11, pp.~3309--3320, 2006.

\bibitem{xiang2018output}
W.~Xiang, H.-D. Tran, and T.~T. Johnson, ``Output reachable set estimation and verification for multilayer neural networks,'' {\em IEEE transactions on neural networks and learning systems}, vol.~29, no.~11, pp.~5777--5783, 2018.

\bibitem{lahijanian2015formal}
M.~Lahijanian, S.~B. Andersson, and C.~Belta, ``Formal verification and synthesis for discrete-time stochastic systems,'' {\em IEEE Transactions on Automatic Control}, vol.~60, no.~8, pp.~2031--2045, 2015.

\bibitem{reinhart2017hybrid}
R.~F. Reinhart, Z.~Shareef, and J.~J. Steil, ``Hybrid analytical and data-driven modeling for feed-forward robot control,'' {\em Sensors}, vol.~17, no.~2, p.~311, 2017.

\bibitem{reinhart2011neural}
R.~F. Reinhart and J.~J. Steil, ``Neural learning and dynamical selection of redundant solutions for inverse kinematic control,'' in {\em 2011 11th IEEE-RAS International Conference on Humanoid Robots}, pp.~564--569, IEEE, 2011.

\bibitem{LASA}
M.~Khansari, E.~Klingbeil, and O.~Khatib, ``Adaptive human-inspired compliant contact primitives to perform surface--surface contact under uncertainty,'' {\em The International Journal of Robotics Research}, vol.~35, no.~13, pp.~1651--1675, 2016.

\bibitem{pan2016model}
H.~Pan, Y.~Li, Y.~Cao, and Z.~Ma, ``Model checking computation tree logic over finite lattices,'' {\em Theoretical computer science}, vol.~612, pp.~45--62, 2016.

\end{thebibliography}

\end{document}